\documentclass[11pt]{article}
\usepackage{fullpage}
\usepackage{amsmath,amsfonts,amsthm,amssymb}
\usepackage{url}
\usepackage{algorithm}
\usepackage{algorithmic}
\usepackage{bbm}
\usepackage{float}
\usepackage{framed}
\usepackage{enumerate}
\usepackage{enumitem}

\usepackage{color}
\usepackage[usenames,dvipsnames,svgnames,table]{xcolor}
\usepackage[colorlinks=true,
            linkcolor=red,
            urlcolor=blue,
            citecolor=gray]{hyperref}

\DeclareMathOperator*{\E}{\mathbb{E}}
\let\Pr\relax
\DeclareMathOperator*{\Pr}{\mathbb{P}}
\DeclareMathOperator{\tr}{tr}
\DeclareMathOperator{\rank}{rank}
\DeclareMathOperator{\nnz}{nnz}
\DeclareMathOperator{\poly}{poly}
\newcommand{\R}{\mathbb{R}}

\makeatletter
\newtheorem*{rep@theorem}{\rep@title}
\newcommand{\newreptheorem}[2]{%
\newenvironment{rep#1}[1]{%
 \def\rep@title{#2 \ref{##1}}%
 \begin{rep@theorem}}%
 {\end{rep@theorem}}}
\makeatother

\newtheorem{theorem}{Theorem}
\newtheorem*{theorem*}{Theorem}
\newreptheorem{theorem}{Theorem}

\newtheorem{lemma}{Lemma}
\newreptheorem{lemma}{Lemma}
\newtheorem*{lemma*}{Lemma}

\newcommand{\eqdef}{\mathbin{\stackrel{\rm def}{=}}}
\newcommand{\norm}[1]{\|#1\|}
\newcommand{\bs}[1]{\boldsymbol{#1}}
\newcommand{\bv}[1]{\mathbf{#1}}

\newcommand{\algoname}[1]{\textnormal{\textsc{#1}}}

\title{Uniform Sampling for Matrix Approximation}

\author{Michael B. Cohen, Yin Tat Lee, Cameron Musco\\ Christopher Musco, Richard Peng, Aaron Sidford \\\\
Massachusetts Institute of Technology\\
\{micohen,yintat,cpmusco,cnmusco,rpeng,sidford\}@mit.edu
}

\date{}

\begin{document}
\maketitle

\begin{abstract}


Random sampling has become a critical tool in solving massive matrix problems. For linear regression, a small, manageable set of data rows can be randomly selected to approximate a tall, skinny data matrix, improving processing time significantly. For theoretical performance guarantees, each row must be sampled with probability proportional to its \emph{statistical leverage score}. Unfortunately, leverage scores are difficult to compute. A simple alternative is to sample rows uniformly at random. While this often works, uniform sampling will eliminate critical row information for many natural instances. 

We take a fresh look at uniform sampling by examining what information it \emph{does} preserve.
Specifically, we show that uniform sampling yields a matrix that, in some sense, well approximates a large fraction of the original. While this weak form of approximation is not enough for solving linear regression directly, it \emph{is} enough to compute a better approximation.

This observation leads to simple iterative row sampling algorithms for matrix approximation that run in input-sparsity time and preserve row structure and sparsity at all intermediate steps.
In addition to an improved understanding of uniform sampling, our main proof introduces a structural result of independent interest: we show that every matrix can be made to have low coherence by reweighting a small subset of its rows. 

\end{abstract}

\thispagestyle{empty}
\clearpage
\setcounter{page}{1}

\section{Introduction}
\label{sec:intro}
Many fast, randomized algorithms for solving massive regression problems rely on the fundamental building block of \emph{spectral approximation}. For a tall, narrow data matrix $\bv{A}$, these methods find a shorter approximate data matrix, $\tilde{\bv{A}}$, such that, for all vectors $\bv{x}$, $\|\bv{\tilde A}\bv{x}\|_2 \approx \|\bv{A}\bv{x}\|_2$.
A recent explosion in work on this problem has lead to extremely fast algorithms, all of which rely on variations of Johnson-Lindenstrauss random projections \cite{clarkson2013low, meng2013, osnap, peng_iterative}. 

%

By re-examining uniform sampling, 
a heuristic known to work for \emph{low coherence} data,
we give spectral approximation algorithms that avoid projection entirely.  Our methods are the first to match state-of-the-art runtimes while preserving row structure and sparsity in all matrix operations.  

It is known that for a data matrix $\bv{A} \in \R^{n\times d}$, a spectral approximation can be obtained by sampling $O(d\log d)$ rows, each with probability proportional to its \emph{statistical leverage score} \cite{Drineas:2006:SAL:1109557.1109682,Spielman:2008:GSE:1374376.1374456}. The leverage score of $\bv{A}$'s  $i^\text{th}$ row, $\bv{a}_i$, is $\tau_i = \bv{a}_i^\top(\bv{A}^\top\bv{A})^+\bv{a}_i$. A higher leverage score indicates that $\bv{a}_i$ is more important in composing the spectrum of $\bv{A}$.

Unfortunately, leverage scores are difficult to calculate -- finding them involves computing $(\bv{A}^\top\bv{A})^+$, which is as slow as solving our regression problem in the first place! In practice, data is often assumed to have low coherence \cite{DBLP:journals/jmlr/MohriT11}, in which case simply selecting rows uniformly at random works
\cite{avron2010blendenpik, Kumar:2012}. However, uniform sampling  could be disastrous -- if $\bv{A}$ contains a row with some component orthogonal to all other rows, removing it will reduce the rank of $\bv{A}$ and thus we cannot possibly preserve all vector products ($\|\bv{\tilde A}\bv{x}\|_2$ will start sending some vectors to $0$). Any uniform sampling scheme is likely to drop any such single row.\footnote{When leverage score sampling, such a row would have the highest possible leverage score of 1.} 

Possible fixes include randomly ``mixing'' data points to avoid degeneracies \cite{avron2010blendenpik}. However, this approach sacrifices sparsity and structure in our data matrix, increasing storage and runtime costs. Is there a more elegant fix? First note that sampling $\bv{A}$ by \emph{approximate} leverage scores is fine, but we may need to select more than the optimal $O(d\log d)$ rows. With that in mind, consider the following straightforward algorithm for iterative sampling (inspired by \cite{peng_iterative}):

%
%
%
%

\begin{enumerate}[leftmargin=1.5cm, label=\bfseries Step \arabic*]
	\item Reduce $\bv{A}$ significantly by sampling uniformly.
	\label{proto:uniform_sample}
	\item Approximate $(\bv{A}^\top\bv{A})^+$ using the smaller matrix and estimate leverage scores for $\bv{A}$.
	\label{proto:approx}
	\item Resample rows from $\bv{A}$ using these estimates, obtaining a spectral approximation $\bv{\tilde A}$.
	\label{proto:resample}
	\item Repeat from \textbf{Step 1} to reduce $\bv{\tilde A}$ further and obtain a smaller approximation.
	\label{proto:repeat}
\end{enumerate}

While intuitive, this scheme was not previously known to work! Our main technical result is proving that it does. This process (and related schemes) will quickly converge on a small spectral approximation to $\bv{A}$ -- i.e. with $O(d\log d)$ rows.

A few results come close to an analysis of such a routine -- in particular, two iterative sampling schemes are analyzed in \cite{peng_iterative}. However, the first ultimately requires Johnson-Lindenstrauss projections that mix rows, something we were hoping to avoid. The second almost maintains sparsity and row structure (except for possibly including rows of the identity in $\bv{\tilde A}$), but its convergence rate depends on the condition number of $\bv{A}$.

More importantly, both of these results are similar in that they rely on the primitive that a (possibly poor) spectral approximation to $\bv{A}$ is sufficient for approximately computing leverage scores, which are in turn good enough for obtaining an even better spectral approximation. As mentioned, uniform sampling will not in general give a spectral approximation -- it does not preserve information about all singular values. Our key contribution is a better understanding of what information uniform sampling \emph{does} preserve. It turns out that, although weaker than a spectral approximation, the matrix obtained from uniform sampling can nonetheless give leverage score estimates that are good enough to obtain increasingly better approximations to $\bv{A}$. 

\subsection{Our Approach}
\label{our_approach}

Suppose we compute a set of leverage score estimates, $\{\tilde{\tau}_i\}$, using $(\bv{\tilde A}^\top\bv{\tilde A})^+$ in place of $(\bv{A}^\top\bv{A})^+$ for some already obtained matrix approximation $\tilde{\bv{A}}$. As long as our leverage score approximations are \emph{upper bounds} on the true scores ($\tilde{\tau}_i \geq \tau_i$) we can use them for sampling and still obtain a spectral approximation to $\bv{A}$ \cite{peng_iterative}. 
However, the number of samples we take will increase to
\begin{align*}
c\cdot\log d \cdot \sum_{i=1}^n \tilde{\tau}_i
\end{align*}
where $c$ is some fixed constant. Note that, when sampling by exact leverage scores, it can be shown that $\sum_{i=1}^n \tau_i \leq d$ so we take $O(d\log d)$ rows.

Thus, to prove that our proposed iterative algorithm works, we need to show that, if we uniformly sample a relatively small number of rows from $\bv{A}$ (\ref{proto:uniform_sample}) and estimate leverage scores using these rows (\ref{proto:approx}), then the sum of our estimates will be small. Then, when we sample by these estimated leverage scores in \ref{proto:resample}, we can sufficiently reduce the size of $\bv{A}$. Note that we will \emph{not} aim to reduce $\bv{A}$ to $O(d\log d)$ height in one shot -- we just need our leverage estimates to sum to say, $n/(2c\log d)$, which allows us to cut the large matrix in half at each step.

In prior work, the sum of overestimates was bounded by estimating \emph{each} leverage score to within a multiplicative factor. This requires a spectral approximation, which is why previous iterative sampling schemes could only boost poor spectral approximations to better spectral approximations. Of course, a ``for each'' statement is not required, and we will not get one through uniform sampling. Thus, our core result avoids this technique.
Specifically, we show,

\begin{theorem}[Leverage Score Approximation via Uniform Sampling]\label{main_lemma} 
For any $m$, we can select $O(m)$ rows uniformly at random from $\bv{A}$ to obtain $\bv{\tilde A}$. Then, letting $\{ \tilde \tau_i \}$ be a set of leverage score estimates for $\bv{A}$ computed using $\bv{\tilde A}$\footnote{We decribe exactly how each $\tilde \tau_i$ is computed when we prove Theorem \ref{main_lemma} in Section \ref{alternative}.}, both of the following hold:
\begin{align*}
\forall i, ~\tilde \tau_i \ge \tau_i,
\end{align*}
\begin{align*}
\E \left [ \sum_{i=1}^n \tilde \tau_i \right ]  \le \frac{nd}{m}.
\end{align*}
\end{theorem}

The validity of our proposed iterative sampling scheme immediately follows from Theorem \ref{main_lemma}. For example, if we set $m = O(d\log d)$ with a high enough constant, $c\log d \sum \tilde \tau_i \le \frac{n}{2}$, allowing us to cut our matrix in half. Alternatively, if we uniformly sample $m = O(n)$ rows (say n/2) then $c\log d \sum \tilde \tau_i \le O(d\log d)$, so we can cut our matrix down to $O(d\log d)$ rows. There is a convenient tradeoff -- the more rows uniformly sampled in \ref{proto:uniform_sample}, the more we can cut $\bv{A}$ down by in \ref{proto:resample}. This tradeoff leads to natural recursive and iterative algorithms for row sampling.

We give a proof of Theorem \ref{main_lemma} using a clean expectation argument. By considering the estimated leverage score for a row computed using our uniformly sampled matrix \emph{unioned} with that row itself, we can bound $\E \tilde \tau_i$ for all $i$. 

Although not stated here for conciseness, we also prove versions of Theorem \ref{main_lemma} with slightly different guarantees (Theorems \ref{leverage_score_undersampling} and  \ref{uniform_stronger}) using a technique that we believe is of independent interest. It is well known that, if $\bv{A}$ has low coherence -- that is, has a low maximum leverage score -- then uniform sampling from the matrix is actually sufficient for obtaining a full spectral approximation. The uniform rate will upper bound the leverage score rate for every row. With this in mind, we show a powerful fact: while many matrices do not have low coherence, for any $\bv{A}$, we can decrease the weight on a small subset of rows to make the matrix have low coherence. Specifically,

\begin{lemma}[Coherence Reducing Reweighting]
\label{coherence_reweighting_intial}
For any $n \times d$ matrix $\bv{A}$ and any coherence upper bound $\alpha > 0$ there exists a diagonal reweighting matrix $\bv{W} \in \mathbb{R}^{n\times n}$ with all entries in $[0, 1]$ and just $(d/\alpha)$ entries not equal to 1, such that:
\begin{align*}
\forall i, ~\tau_i(\bv{WA}) \le \alpha.
\end{align*}
\end{lemma}
Intuitively, this lemma shows that uniform sampling gives a matrix that spectrally approximates a large sub-matrix of the original data. It follows from our more general Theorem \ref{weighting_existance}, which describes exactly how leverage scores of $\bv{A}$ can be manipulated through row reweighting.

We never actually find $\bv{W}$ explicitly -- simply its existence implies our uniform sampling theorems! 
As explained, since $\bv{WA}$ has low coherence, uniform sampling would give a spectral approximation to the reweighted matrix and thus a multiplicatively good approximation to \emph{each} leverage score. Thus, the sum of estimated leverage scores for $\bv{WA}$ will be low, i.e. $< O(d)$. It can be shown that, for any row that is not reweighted, the leverage score in $\bv{A}$ computed using a uniformly sampled $\bv{\tilde A}$, is never greater than the corresponding leverage score in $\bv{WA}$ computed using a uniformly sampled $\bv{\widetilde{WA}}$.
Thus, the sum of approximate leverage scores for rows in $\bv{A}$ that are not reweighted is small by comparison to their corresponding leverage scores in $\bv{WA}$. How about the rows that \emph{are} reweighted in $\bv{WA}$? Lemma \ref{coherence_reweighting_intial} claims there are not too many of these -- we can trivially bound their leverage score estimates by 1 and even then the total sum of estimated leverage scores will be small. 

This argument gives the result we need: even if a uniformly sampled $\tilde{\bv{A}}$ cannot be used to obtain good per row leverage score upper bounds, it is sufficient for ensuring that the sum of all leverage score estimates is not too high. 
%
%
%

\subsection{Road Map}
\begin{description}
\item[Section \ref{background}]  Survey prior work on randomized linear algebra and spectral matrix approximation.
\item[Section \ref{notation}] Review frequently used notation and important foundational lemmas.
\item[Section \ref{alternative}] Prove that uniform sampling is sufficient for leverage score estimation (Theorem \ref{main_lemma}).
\item[Section \ref{reweighting}] Show the existence of small, coherence-reducing reweightings (Theorem \ref{weighting_existance}, Lemma \ref{coherence_reweighting_intial}).
\item[Section \ref{undersampling}] Use this result to prove alternative versions of Theorem \ref{main_lemma} (Theorems \ref{leverage_score_undersampling} and \ref{uniform_stronger}).
\item[Section \ref{sec:algo}] Describe simple and efficient iterative algorithms for spectral matrix approximation.
\end{description}

\section{Background}\label{background}
\subsection{Randomized Numerical Linear Algebra}

In the past decade, fast randomized algorithms for matrix problems have risen to prominence. Numerous results give improved running times for matrix multiplication, linear regression, and low rank approximation --  
helpful surveys of this work include \cite{Mahoney:2011} and \cite{Halko:2011}. 
In addition to asymptotic runtime gains, randomized alternatives to standard linear algebra tools tend to offer significant gains in terms of data access patterns and required working memory. 

Algorithms for randomized linear algebra often work by generically reducing problem size -- large matrices are compressed (using randomness) to smaller approximations which are processed deterministically via standard linear algebraic methods. 
Methods for matrix reduction divide roughly into two categories -- random projection methods \cite{Clarkson:2009,clarkson2013low, meng2013, osnap, Sarlos06} and sampling methods \cite{Drineas:2004:FMC1, Drineas:2006:FMC2, Drineas:20063, Drineas:2008:RCM:1461865.1461889, DBLP:journals/corr/abs-1005-3097, peng_iterative}. 

Random projection methods recombine rows or columns from a large matrix to form a much smaller problem that approximates the original. Descending from the Johnson-Lindenstrauss Lemma \cite{citeulike:7030987} and related results, these algorithms are impressive for their simplicity and speed -- reducing a large matrix simply requires multiplication by an appropriately chosen random matrix. 

Sampling methods, on the other hand, seek to approximate large matrices by judiciously selecting (and reweighting) few rows or columns. Sampling itself is even simpler and faster than random projection -- the challenge becomes efficiently computing the correct measure of ``importance'' for rows or columns. More important rows or columns are selected with higher probability. 

\subsection{Approximate Linear Regression}
We focus on linear regression, i.e. solving overdetermined systems, which requires our matrix reduction step to produce a spectral approximation $\bv{\tilde A}$ to the data matrix $\bv{A}$.  One possibility is to obtain a $(1\pm \epsilon)$ approximation with $O(d\log d/\epsilon^2)$ rows, and to solve regression on the smaller problem to give an approximate solution.
To improve stablility and achieve $\log (1/\epsilon)$ dependence, randomized schemes can be combined with known iterative regression algorithms. These methods only require a constant factor spectral approximation with $O(d\log d)$ rows and are addressed in \cite{avron2010blendenpik, DBLP:journals/siamsc/CoakleyRT11, clarkson2013low, DBLP:journals/siamsc/MengSM14, rokhlinTygert}. 

When random projections are used, $\tilde{\bv{A}} = \bv{\Pi}\bv{A}$ for some randomly generated matrix $\bv{\Pi}$
which is known as a \emph{subspace embedding}. 
Work on subspace embeddings goes back to \cite{Papadimitriou:1998} and \cite{Sarlos06}.
Recent progress has significantly sped up the process of computing $\bv{\Pi A}$, leading to the first \emph{input-sparsity time} algorithms for linear regression (or nearly input-sparsity time if iterative methods are employed) \cite{clarkson2013low, meng2013, osnap}.

\subsection{Row Sampling}

As discussed, an alternative route to spectral matrix approximation is importance sampling. Specifically, $O(d\log d/\epsilon^2)$ rows can be sampled with probability proportional to their leverage scores, as suggested in \cite{Drineas:2006:SAL:1109557.1109682} and proved by Spielman and Srivastava \cite{Spielman:2008:GSE:1374376.1374456}. 
Spielman and Srivastava were specifically focused on spectral approximations for the edge-vertex incidence matrix of a graph. This is more commonly referred to as \emph{spectral sparsification} -- a primitive that has been very important in literature on graph algorithms. Each row in a graph's (potentially tall) edge-vertex incident matrix corresponds to an edge and the row's leverage score is exactly the edge's \emph{weighted effective resistance}, whicth is used as the sampling probability in \cite{Spielman:2008:GSE:1374376.1374456}. 

This application illustrates an important point: when applied to spectral graph sparsification, it is critical that the $\bv{A}$ is compressed via sampling instead of random projection. Sampling ensures that $\bv{\tilde A}$ contains only reweighted rows from $\bv{A}$ -- i.e. it remains an edge-vertex incidence matrix. In general, sampling is interesting because it preserves row structure. Even if that structure is just a certain level of sparsity, it can reduce memory requirements and accelerate matrix operations. 

While leverage scores for the edge-vertex incidence matrix of a graph can be computed quickly \cite{5671167, Spielman:2004}, in general, computing leverage scores requires evaluating $(\bv{A}^\top\bv{A})^+$, which is as difficult as solving regression in the first place. Li, Miller, and Peng address this issue with methods for iteratively computing good row samples \cite{peng_iterative}. Their algorithms achieve input-sparsity time regression, but are fairly involved and rely on intermediate operations that ultimately require Johnson-Lindenstrauss projections, mixing rows and necessitating dense matrix operations. An alternative approach from \cite{peng_iterative} does preserve row structure (except for possible additions of rows from the identity to intermediate matrices) but converges in a number of steps that depends on $\bv{A}$'s condition number.

\section{Notation and Preliminaries}\label{notation}

\subsection{Singular Value Decomposition and Pseudoinverse}

For any $\bv{A} \in \mathbb{R}^{n \times d}$ with rank $r$, we write the reduced singular value decomposition, $\bv{A}=\bv{U}\bv{\Sigma}\bv{V}^{\top}$. $\bv{U} \in \mathbb{R}^{n\times r}$ and $\bv{V} \in \mathbb{R}^{d \times r}$ have orthonormal columns and $\bv{\Sigma} \in \mathbb{R}^{r \times r}$ is diagonal and contains the nonzero singular values of $\bv{A}$. $\bv{A}^\top\bv{A} = \bv{V}\bv{\Sigma}\bv{U}^\top\bv{U}\bv{\Sigma}\bv{V}^{\top} = \bv{V}\bv{\Sigma}^2\bv{V}^{\top}$. Let $(\bv{A}^\top \bv{A})^+$ denote the Moore-Penrose pseudoinverse of $\bv{A}^\top \bv{A}$. $(\bv{A}^\top \bv{A})^+ = \bv{V}(\bv{\Sigma}^{-1})^2\bv{V}^{\top}$.

\subsection{Spectral Approximation}
\label{spectral_approximation_background}
For any $\lambda \ge 1$, we say that $\bv{\tilde A} \in \mathbb{R}^{n' \times d}$ is a $\lambda$-\emph{spectral approximation} of $\bv{A} \in \mathbb{R}^{n \times d}$ if, $\forall \bs{x} \in \mathbb{R}^d$
\begin{align}
\frac{1}{\lambda} \norm{\bv{Ax}}^2 &\le \norm{\bv{\tilde Ax}}^2 \le \norm{\bv{Ax}}^2, \text{ or equivalently} \nonumber\\
\frac{1}{\lambda} \bv{x}^\top \bv{A}^\top \bv{A} \bv{x} &\le \bv{x}^\top \bv{\tilde A}^\top \bv{\tilde A} \bv{x} \le \bv{x}^\top \bv{A}^\top \bv{A} \bs{x}.
\end{align}
Letting $\sigma_i$ denote the $i^\text{th}$ singular value of a matrix, $\lambda$-spectral approximation implies:
\begin{align*}
\forall i, ~\frac{1}{\lambda} \sigma_i(\bv{A}) \le\sigma_i(\bv{\tilde A})  \le \sigma_i(\bv{A}) .
\end{align*}
%
So, a spectral approximation preserves the magnitude of matrix-vector multiplication with $\bv{A}$, the value of $\bv{A}^\top \bv{A}$'s quadratic form, and consequently, each singular value of $\bv{A}$. For conciseness, we sometimes write $\frac{1}{\lambda}\bv{A}^\top \bv{A} \preceq \bv{\tilde A}^\top \bv{\tilde A} \preceq \bv{A}^\top \bv{A}$ where $\bv{C} \preceq \bv{D}$ indicates that $\bv{D} - \bv{C}$ is positive semidefinite. Even more succinctly, $\bv{\tilde A}^\top \bv{\tilde A} \approx_\lambda \bv{A}^\top \bv{A}$ denotes the same condition. 

\subsection{Leverage Scores}

The leverage score of the $i^{\text{th}}$ row $\bv{a}_i^\top$ of $\bv{A}$ is: 
\begin{align}
\tau_i(\bv{A}) \eqdef \bv{a}_i^\top (\bv{A^\top A})^+ \bv{a}_i.
\end{align}
We also define the related \emph{cross leverage score} as
$\tau_{ij}(\bv{A}) \eqdef \bv{a}_i^\top (\bv{A^\top A})^+ \bv{a}_j$.
Let $\bs{\tau}(\bv{A})$ be a vector containing $\bv{A}$'s $n$ leverage scores. $\bs{\tau}(\bv{A})$ is the diagonal of $\bv{A}(\bv{A^\top A})^+ \bv{A}^\top$, which is a projection matrix. Therefore, $\tau_i(\bv{A}) = \bv{\mathbbm{1}}_{i}^\top \bv{A}(\bv{A^\top A})^+ \bv{A}^\top \bv{\mathbbm{1}}_{i} \le 1$. In addition to this individual bound, because $\bv{A}(\bv{A^\top A})^+ \bv{A}^\top$ is a projection matrix, the sum of $\bv{A}$'s leverage scores is equal to the matrix's rank: 
\begin{align}\label{rank_bound}
\sum_{i=1}^n \tau_i(\bv{A}) = \tr(\bv{A}(\bv{A^\top A})^+ \bv{A}^\top) = \sum_{i=1}^n \lambda_i(\bv{A}) = \rank(\bv{A}(\bv{A^\top A})^+ \bv{A}^\top) = \rank(\bv{A}) \le d.
\end{align}
A row's leverage score measures how important it is in composing the row space of $\bv{A}$. If a row has a component orthogonal to all other rows, its leverage score is $1$. Removing it would decrease the rank of $\bv{A}$, completely changing its row space. If all rows are the same, each has leverage score $d/n$. The \emph{coherence} of $\bv{A}$ is $\norm{\bs{\tau}(\bv{A})}_\infty$. If $\bv{A}$ has low coherence, no particular row is especially important. If  $\bv{A}$ has  high coherence, it contains at least one row whose removal would significantly affect the composition of $\bv{A}$'s row space.
A characterization that helps with this intuition follows:
\begin{lemma}\label{leverage_score_opt}
For all $\bv{A} \in \mathbb{R}^{n\times d}$ and $i\in[n]$ we have that 
\begin{align*}
\tau_i(\bv{A}) = \min_{\bv{A}^{\top} \bv{x}=\bv{a}_i} \norm{\bv{x}}^{2}_2.
\end{align*}
Let $\bv{x}_i$ denote the optimal $\bv{x}$ for $\bv{a}_i$. The $j^\text{th}$ entry of $\bv{x}_i$ is given by $\bv{x}_{i}^{(j)}=\tau_{ij}(\bv{A})$.

%
\end{lemma}
\begin{proof}
For the solution $\bv{x}_i$ to have minimal norm, we must have $\bv{x}_i \perp \ker(\bv{A}^{\top})$. Thus, $\bv{x}_i \in \mathrm{im}(\bv{A})$ and we can write $\bv{x}_i =\bv{A} \bv{c}$ for some $\bv{c} \in \mathbb{R}^{d}$. Using the constraints of the optimization problem we have that $\bv{A}^\top \bv{x}_i = \bv{A^{\top} A} \bv{c}=\bv{a}_i$. Thus $\bv{c} = \bv{(A^\top A)}^+ \bv{a}_i$, so $\bv{x}_i =\bv{A(A^\top A)}^+ \bv{a}_i$.
This gives $x_i^{(j)} = \bv{a}_j^\top \bv{(A^\top A)}^+ \bv{a}_i = \tau_{ij}(\bv{A})$. Furthermore: 
\begin{align*}
\norm{\bv{x}_i}_2^2 &= \bv{a}_i^\top \bv{(A^\top A)}^+ \bv{A^\top A} \bv{(A^\top A)}^+ \bv{a}_i\\
&= \bv{a}_i^\top \bv{(A^\top A)}^+ \bv{a}_i = \tau_i(\bv{A}).
\end{align*}
\end{proof}

We often approximate the leverage scores of $\bv{A}$ by computing them with respect to some other matrix $\bv{B} \in \mathbb{R}^{n' \times d}$. We define the \emph{generalized leverage score}:
\begin{align}
\tau^{\bv{B}}_i(\bv{A}) \eqdef
\begin{cases}
 \bv{a}_i^\top(\bv{B^\top B)^+} \bv{a}_i & \text{if } \bv{a}_i\perp\ker(\bv{B}),\\
\infty & \text{otherwise}.
\end{cases}
\end{align}

If $\bv{a}_i$ has an component in $\ker(\bv{B})$, we set its generalized leverage score to $\infty$, since it might be the only row in $\bv{A}$ pointing in this direction. Thus, when sampling rows, we cannot remove it. We could set the generalized leverage score to $1$, but using $\infty$ simplifies notation in some of our proofs. If $\bv{B}$ is a spectral approximation for $\bv{A}$, then every generalized leverage score is a good multiplicative approximation to its corresponding true leverage score:

\begin{lemma}[Leverage Score Approximation via Spectral Approximation - Lemma 4.3 of \cite{peng_iterative}]\label{leverage_score_approx} 
If $\bv{B}$ is a $\lambda$-spectral approximation of $\bv{A}$, so $\frac{1}{\lambda} \bv{A^\top A} \preceq \bv{B^\top B} \preceq \bv{A^\top A}$, then $\tau_i(\bv{A}) \le \tau_i^{\bv{B}}(\bv{A}) \le \lambda \cdot \tau_i(\bv{A})$.
\end{lemma}
\begin{proof}
This follows simply from the definition of leverage scores and generalized leverage scores and the fact that $\frac{1}{\lambda} \bv{A^\top A} \preceq \bv{B^\top B} \preceq \bv{A^\top A}$ implies $\lambda (\bv{A^\top A})^+ \succeq (\bv{B^\top B})^+ \succeq \bv{A^\top A}$.
\end{proof}


\subsection{Leverage Score Sampling}

Sampling rows from $\bv A$ according to their exact leverage scores
gives a spectral approximation for $\bv A$ with high probability. Sampling by leverage score overestimates
also suffices. Formally:  

\begin{lemma}[Spectral Approximation via Leverage Score Sampling]\label{leverage_score_sampling} Given an error parameter $0 < \epsilon < 1$, let $\bv u$ be a vector of leverage score overestimates, i.e., $\tau_{i}(\bv A)\le u_{i}$ for all $i$. Let $\alpha$ be a sampling rate parameter and let $c$ be a fixed positive constant. For each row, we define a sampling probability $p_i = \min \{1,\alpha\cdot u_i  c \log d\}$. Furthermore, let $\mathtt{Sample}(\bv u,\alpha)$ denote a function which returns a random diagonal matrix $\bv{S}$ with independently
chosen entries. $\bv{S}_{ii} = \frac{1}{\sqrt{p_i}}$ with probability $p_i$ and $0$ otherwise.

If we set $\alpha = \epsilon^{-2}$, $\bv S = \mathtt{Sample}(\bv u,\epsilon^{-2})$ has at most $\sum_i \min \{1,\alpha\cdot u_i  c \log d \} \le \alpha c\log d\norm{\bv{u}}_1$
non-zero entries and $\frac{1}{\sqrt{1+\epsilon}} \bv{SA}$ is a ${\frac{1+\epsilon}{1-\epsilon}}$-spectral approximation for $\bv A$ with probability at least $1-d^{-c/3}$. 
\end{lemma}
\noindent For completeness, we prove Lemma \ref{leverage_score_sampling} in Appendix \ref{sec:app:lever_scores} using a matrix concentration result of \cite{tropp2012user}.


\section{Leverage Score Estimation via Uniform Sampling}\label{alternative}


In this section, we prove Theorem \ref{main_lemma} using a simple expectation argument. We restate a more complete version of the theorem below:
\begin{reptheorem}{main_lemma}[Full Statement]
Given any $\bv{A} \in \mathbb{R}^{n \times d}$. Let $S$ denote a uniformly random sample of $m$ rows from $\bv{A}$ and let $\bv{S} \in \mathbb{R}^{n \times n}$ be its diagonal indicator matrix (i.e. $\bv{S}_{ii} =1$ for $i \in S$, $\bv{S}_{ii} = 0$ otherwise). Define 
\begin{align*}
\tilde \tau_i \eqdef \begin{cases}
\tau^{\bv{SA}}_i\left(\bv{A}\right) & \text{if } i \in S,\\
\frac{1}{1 + \frac{1}{ \tau^{\bv{SA}}_i\left(\bv{A}\right)}} & \text{otherwise}.
\end{cases}
\end{align*}
Then, $\tilde \tau_i \ge \tau_i(\bv{A})$ for all $i$ and
\begin{equation*}
\E \left [ \sum_{i=1}^n \tilde \tau_i \right ] \leq \frac{nd}{m}.
\end{equation*}
\end{reptheorem}

\begin{proof}
First we show that our estimates are valid leverage score upper bounds, i.e. $\tilde \tau_i \ge \tau_i(\bv{A})$.
Let $\bv{S}^{(i)}$ be the diagonal indicator matrix for $S \cup \{ i \}$. We claim that, for all $i$,
\begin{align}
\label{s_union_i}
\tilde \tau_i = \tau^{\bv{S}^{(i)}\bv{A}}_i\left(\bv{A}\right).
\end{align}
This is proved case-by-case:
\begin{enumerate}
\item  When $i \in S$, $\bv{S} = \bv{S}^{(i)}$ so it holds trivially. 
\item When $i \notin S$ and $\bv{a}_i \not \perp \ker(\bv{SA})$, then by definition, $\tau^{\bv{SA}}_i(\bv{A}) = \infty$ and $\tilde \tau_i = \frac{1}{1 + \frac{1}{ \infty}} = 1 = \tau^{\bv{S}^{(i)}\bv{A}}_i(\bv{A})$.
\item When $i \notin S$ and $\bv{a}_i \perp \ker(\bv{SA})$ then by the Sherman-Morrison formula for pseudoinverses \cite[Thm 3]{meyer1973generalized},
\begin{align*}
\tau^{\bv{S}^{(i)}\bv{A}}_i(\bv{A}) & = \bv a_{i}^{\top}\left(\bv{A}^{\top}\bv{S}^2\bv{A} + \bv a_{i}\bv a_{i}^{\top}\right)^{+}\bv a_{i}\\
 &=\bv a_{i}^{\top}\left(\left(\bv{A}^{\top}\bv{S}^2 \bv{A}\right)^{+}
 - \frac{\left(\bv{A^{\top}}\bv{S}^2\bv{A}\right)^{+}
 \bv a_{i}\bv a_{i}^{\top}\left(\bv{A^{\top}}\bv{S}^2 \bv{A}\right)^{+}}
 {1+\bv a_{i}^{\top}\left(\bv{A^{\top}}\bv{S}^2\bv{A}\right)^{+}\bv a_{i}}\right)\bv a_{i}
 \tag{Sherman-Morrison formula}\\
 & = \tau^{\bv{SA}}_i\left(\bv{A}\right) - \frac{\tau^{\bv{SA}}_i\left(\bv{A}\right)^{2}}{1+\tau^{\bv{SA}}_i(\bv{A})}
 =  \frac{1}{1 + \frac{1}{ \tau^{\bv{SA}}_i\left(\bv{A}\right)}}
 =  \tilde \tau_i.
\end{align*}
\end{enumerate}
By \eqref{s_union_i} and the fact that $\bv{A^\top}{\bv{S}^{(i)}}^2\bv{A} \preceq \bv{A^\top A}$  (see Lemma \ref{leverage_score_approx}), we have $\tilde \tau_i = \tau^{\bv{S}^{(i)}\bv{A}}_i(\bv{A}) \ge \tau_i(\bv{A})$, so our estimates are upper bounds as desired. It remains to upper bound the expected sum of $\tilde \tau_i$.
We can break down the sum as
\begin{equation*}
\sum_{i = 1}^{n} \tilde \tau_i = \sum_{i \in S} \tilde \tau_i + \sum_{i \not \in S} \tilde \tau_i.
\end{equation*}
The first term is simply the sum of $\bv{SA}$'s leverage scores, so it is equal to $rank(\bv{SA}) \le d$ by \eqref{rank_bound}. To bound the second term, consider a random process that first selects $\bv{S}$, then selects a random row $i$ $\not \in S$ and returns $\tilde \tau_i$.  There are always exactly $n - m$ rows $\not \in S$, so the value returned by this random process is, in expectation, exactly equal to $\frac{1}{n - m} \cdot \E \sum_{i \not \in S} \tilde \tau_i$.

This random process is \emph{also} equivalent to randomly selecting a set $S'$ of $m+1$ rows, then randomly choosing a row $i \in \bv{S'A}$ and returning its leverage score!
In expectation it is therefore equal to the average leverage score in $\bv{S'A}$. $\bv{S'A}$ has $m+1$ rows and its leverage scores sum to its rank, so we can bound its average leverage score by $\frac{d}{m+1}$.  Overall we have:
\begin{align*}
\E \left [ \sum_{i=1}^n \tilde \tau_i \right ] \leq d + (n-m) \cdot \frac{d}{m+1} \le \frac{d (n+1)}{m+1} \le \frac{nd}{m}.
\end{align*}
\end{proof}

\section{Coherence Reducing Reweighting}\label{reweighting}


In this section, we prove Theorem \ref{weighting_existance}, which shows that we can reweight a small number of rows in any matrix $\bv{A}$ to make it have low coherence. This structural result may be of independent interest. It is also fundamental in proving Theorem \ref{leverage_score_undersampling}, a slightly stronger and more general version of Theorem \ref{main_lemma} that we will prove in Section \ref{undersampling}. 
 
Actually,  for Theorem \ref{weighting_existance} we prove a more general statement, studying how to select a diagonal row reweighting matrix $\bv{W}$ to arbitrarily control the leverage scores of $\bv{WA}$. 
One simple conjecture would be that, given a vector $\bv{u}$, there always exists a $\bv{W}$ such that $\tau_i(\bv{WA}) = u_i$. This conjecture is unfortunately not true - if $\bv{A}$ is the identity matrix, then $\tau_i(\bv{WA}) = 0$ if $\bv{W}_{ii} = 0$ and $\tau_i(\bv{WA}) = 1$ otherwise. Instead, we show the following:

\begin{theorem}[Leverage Score Bounding Row Reweighting]\label{weighting_existance}
For any $\bv{A} \in \R^{n \times d}$ and any vector $\bv{u} \in \mathbb{R}^n$ with $u_i > 0$ for all $i$, there exists a diagonal matrix $\bv{W} \in \mathbb{R}^{n\times n}$ with $\bv{0} \preceq \bv{W} \preceq \bv{I}$ such that:
\begin{align}\label{levUpper}
\forall i,~\tau_i\left(\bv{WA}\right) \le u_i,
\end{align}
and
\begin{align}\label{levSum}
\sum_{i: \bv{W}_{ii} \neq 1} u_i\le d.
\end{align}
\end{theorem}

Note that \eqref{levUpper} is easy to satisfy -- it holds if we set $\bv{W} = \bv{0}$. Hence, the main result is the second claim . Not only does a $\bv{W}$ exist that gives the desired leverage score bounds, but it is only necessary to reweight rows in $\bv{A}$ with a low total weight in terms of $\bv{u}$.

For any incoherence parameter $\alpha$, if we set $u_i = \alpha$ for all $i$, then this theorem shows the existence of a reweighting that reduces coherence to $\alpha$. Such a reweighting has $\sum_{i: \bv{W}_{ii} \neq 1} \alpha \le d$ and therefore $\left | \{ i: \bv{W}_{ii} \neq 1 \} \right | \le \frac{d}{\alpha}$. So, we see that Lemma \ref{coherence_reweighting_intial} follows as a special case of Theorem \ref{weighting_existance}.

In order to prove Theorem \ref{weighting_existance}, we first give two technical lemmas which are proved in Appendix~\ref{sec:app:lever_scores}. Lemma~\ref{rank_one_updates} describes how the leverage scores of $\bv{A}$ evolve when a single row of $\bv{A}$ is reweighted. We show that, when we decrease the weight of a row, that row's leverage score decreases and the leverage score of all other rows increases.

\begin{lemma}[Leverage Score Changes Under Rank 1 Updates]\label{rank_one_updates} Given any
$\bv A \in \R^{n\times d}$, $\gamma\in(0,1)$, and $i\in[n]$, let $\bv W$ be a diagonal
matrix such that $\bv W_{ii}=\sqrt{1-\gamma}$ and $\bv W_{jj}=1$
for all $j\neq i$. Then, 
\begin{align*}
\tau_{i}\left(\bv{WA}\right)
& =\frac{\left(1-\gamma\right)\tau_{i}\left(\bv A\right)}
{1-\gamma\tau_{i}\left(\bv A\right)}
\le\tau_{i}\left(\bv A\right),
\end{align*}
and for all $j\neq i$,
\begin{align*}
\tau_{j}\left(\bv{WA}\right)
& =\tau_{j}\left(\bv A\right)+\frac{\gamma\tau_{ij}\left(\bv A\right)^{2}}
{1-\gamma\tau_{i}\left(\bv A\right)}
\ge\tau_{j}\left(\bv A\right).
\end{align*}
\end{lemma}

Next we claim that leverage scores are lower semi-continuous in the weighting of the rows. This allows us to reason about weights that arise as the limit of Algorithm \ref{greedy_reweighting} for computing them.

\begin{lemma}[Leverage Scores are Lower Semi-continuous]
\label{lem:lower_semi_tau} $\bs{\tau}(\bv{W}\bv{A})$
is lower semi-continuous in the diagonal matrix $\bv W$, i.e. for any sequence $\bv{W}^{(k)} \rightarrow \overline{\bv W}$
with $\bv W_{ii}^{(k)} \geq 0$ for all $k$ and $i$, we have
\begin{equation*}
\tau_{i}\left(\overline{\bv W}\bv A\right)
\leq\liminf_{k\rightarrow\infty}\tau_{i}\left(\bv W^{(k)}\bv A\right).
\end{equation*}
\end{lemma}
With Lemmas \ref{rank_one_updates} and \ref{lem:lower_semi_tau} in place, we are ready to prove the main reweighting theorem.

\begin{proof}[Proof of Theorem~\ref{weighting_existance}]
We prove the existence of the required $\bv{W}$ by considering the limit of the following algorithm for computing a reweighting matrix. 

\begin{algorithm}[H]
\caption{\algoname{Compute Reweighting} (a.k.a the whack-a-mole algorithm)}
\begin{algorithmic}
\STATE Initialize $\bv{W = I}$
\WHILE{true}
\FOR{$i = 1$ to $n$}
\IF{$\tau_i(\bv{WA}) \ge u_i$}
\IF{$\tau_i(\bv{WA}) < 1$}
\STATE Decrease $\bv{W}_{ii}$ so that $\tau_i(\bv{WA}) = u_i$.
\ELSE
\STATE Set $\bv{W}_{ii} = 0$
\ENDIF
\ENDIF
\ENDFOR
\ENDWHILE
\RETURN{$\bv{W}$}
\end{algorithmic}
\label{greedy_reweighting}
\end{algorithm}

For all $k\geq 0$, let $\bv W^{(k)}$ be the value of $\bv W$ after
the $k^\text{th}$ update to the weight. We show that $\overline{\bv W}=\lim_{k\rightarrow\infty}\bv W^{(k)}$
meets the conditions of Theorem~\ref{weighting_existance}.
First note that Algorithm \ref{greedy_reweighting} is well defined and that all entries
of $\bv W^{(k)}$ are non-negative for all $k\geq0$. To see this,
suppose we need to decrease $\bv W_{ii}^{(k)}$ so that $\tau_{i}(\bv W^{(k+1)}\bv A)= u_{i}$.
Note that the condition $\tau_{i}(\bv W^{(k)}\bv A)<1$ gives
\[
\lim_{\gamma \rightarrow 1} \frac{\left(1-\gamma\right)\tau_{i}\left(\bv W^{(k)}\bv A\right)}
{1-\gamma\tau_{i}\left(\bv W^{(k)}\bv A\right)}
= 0.
\]
Therefore, Lemma \ref{rank_one_updates} shows that we can make $\tau_{i}(\bv W^{(k+1)}\bv A)$
arbitrary small by setting $\gamma$ close enough to $1$. Since the
leverage score for row $i$ is continuous, this implies
that $\bv W^{(k+1)}$ exists as desired.

Since, the entries of $\bv W^{(k)}$ are non-negative and decrease
monotonically by construction, clearly $\overline{\bv W}$ exists.
Furthermore, since setting $\bv W_{ii}=0$ makes $\tau_{i}(\bv W\bv A)=0$,
we see that, by construction, 
\[
\liminf_{k\rightarrow\infty}\tau_{i}\left(\bv W^{(k)}\bv A\right)\leq u_{i}
~\text{for all $i\in[n]$}.
\]
Therefore, by Lemma~\ref{lem:lower_semi_tau}
we have that $\tau_{i}(\overline{\bv W}\bv A)\leq u_{i}$.

It only remains to show that $\sum_{i:\overline{\bv W}_{ii}\neq1} u_{i}\leq d$.
Let $k$ be the first iteration such that $\bv W_{ii}^{(k)}\neq1$
for any $i$ such that $\overline{\bv W}_{ii}\neq1$. Let $S\subseteq[n]$
be the set of rows such that $\bv W_{ii}^{(k)}=0$ and let
$T=\{i:\overline{\bv W}_{ii}\neq1\}-S$. Since decreasing the weight
of one row increases the leverage scores of all other rows, we have
\begin{eqnarray*}
\sum_{i\in T\cup S} u_{i} & \leq & \sum_{i\in T}\tau_{i}\left(\bv W^{(k)}\bv A\right)+\sum_{i\in S}1\\
 & \leq & \rank\left(\bv W^{(k)}\bv A\right)+\left|S\right|.
\end{eqnarray*}
When we set $\bv W_{ii}=0$, it must be the case that $\tau_{i}(\bv W\bv A)=1$.
In this case, removing the $i^\text{th}$ row decreases the rank
of $\bv W\bv A$ by 1 and hence $\rank(\bv W^{(k)}\bv A)\leq d-\left|S\right|$.
Therefore, 
\[
\sum_{i:\overline{\bv W}_{ii}\neq1} u_{i} = \sum_{i\in T\cup S} u_{i}\leq d.
\]

\end{proof}

\section{Leverage Score Approximation via Undersampling}\label{undersampling}
Theorem \ref{main_lemma} alone is enough to prove that a variety of iterative methods for spectral matrix approximation work. However, in this section we prove Theorem \ref{leverage_score_undersampling}, a slight strengthening and generalization that
improves runtime bounds, proves correctness for some alternative sampling schemes, and gives some more intuition for why uniform sampling allows us to obtain leverage score estimates with low total sum.

Theorem \ref{leverage_score_undersampling} relies on Theorem \ref{weighting_existance}, which intuitively shows that a large fraction
of our matrix $\bv A$ has low coherence. Sampling rows uniformly will
give a spectral approximation for this portion of our matrix. Then, since
few rows are reweighted in $\bv{WA}$, even loose upper bounds on
the leverage scores for those rows will allow us to bound the total
sum of estimated leverage scores when we sample uniformly.

Formally, we show an upper bound on the sum of estimated
leverage scores obtained from \emph{undersampling} $\bv A$ according
to any set of leverage score upper bounds. Uniform sampling $\bv A$
can simply be viewed as undersampling $\bv A$ when all we know is
that each leverage score is upper bounded by $1$. That is, in the uniform case, we set $\bv u=\bv 1$.

The bound in Theorem \ref{leverage_score_undersampling} holds with high probability, rather than in expectation like the bound in Theorem \ref{main_lemma}. This gain comes at a cost of requiring our sampling rate to be higher by a factor of $\log d$. At the end of this section we show how the $\log d$ factor can be removed at least in the case of uniform sampling, giving a high probability statement that matches the bound of Theorem \ref{main_lemma}.

\begin{theorem}[Leverage Score Approximation via Undersampling]\label{leverage_score_undersampling}
Let $\bv u$ be a vector of leverage score overestimates, i.e., $\tau_{i}(\bv A)\le u_{i}$
for all $i$. For some undersampling parameter $\alpha\in(0,1]$, let $\bv S'= \sqrt{\alpha}\sqrt{\frac{3}{4}}\cdot\mathtt{Sample}\left(\bv u,9\alpha\right)$. Let $ u^{(new)}_{i}=\min\{\tau_{i}^{\bv{S'A}}(\bv A), u_{i}\}.$
Then, with high probability, $u^{(new)}_{i}$ is a leverage score overestimate, i.e.
$\tau_{i}(\bv A)\le u_{i}^{(new)}$, and
\[
\sum_{i=1}^{n} u^{(new)}_{i}\leq\frac{3d}{\alpha}.
\]
Furthermore, $\bv S'$ has $O\left(\alpha\left\Vert \bv u\right\Vert _{1}\log d\right)$
nonzeros.
\end{theorem}

\begin{proof} Let $\bv S= \sqrt{\frac{3}{4}} \cdot \mathtt{Sample}\left(\bv u, 9\right)$. 
Since $\bv u$ is a set of leverage score overestimates, Lemma \ref{leverage_score_sampling} (with $\epsilon = 1/3$) shows that, with high probability,
\[
\bv A^{\top}\bv S^{2}\bv A\preceq\bv A^{\top}\bv A.
\]

In $\mathtt{Sample}$, when we include a row, we reweight it by $1/\sqrt{p_i}$. For $\bv{S}' = \sqrt{\alpha}\sqrt{\frac{3}{4}} \cdot \mathtt{Sample}\left(\bv u,9\alpha\right)$, we sample at a rate lower by a factor of $\alpha$ as compared with $\bv{S}$, so we weight rows by a factor of $\frac{1}{\sqrt{\alpha}}$ higher. The $\sqrt{\alpha}$ multiplied by $\bv{S}'$ makes up for this difference. Thus, $\bv{S'}$ is equivalent to $\bv{S}$ with some rows removed. Therefore:
\[
\bv A^{\top}\bv S'^{2} \bv A \preceq \bv A^{\top}\bv S^{2}\bv A\preceq\bv A^{\top}\bv A.
\]
So, for all $i$, $\tau_{i}(\bv A) \leq \tau_{i}^{\bv{S'A}}(\bv A)$. By assumption $\tau_{i}(\bv A) \leq u_i$, so
this proves that $\tau_{i}(\bv A)\le u_{i}^{(new)}$.

By Theorem \ref{weighting_existance}, there is a diagonal matrix $\bv W$
such that $\tau_{i}(\bv{WA})\le\alpha u_{i}$ for all $i$ and
$\sum_{i:\bv W_{ii}\neq1}\alpha u_{i}\le d$. For this $\bv W$, using the fact that $ u^{(new)}_i=\min\{\tau_{i}^{\bv{S'A}}(\bv A), u_{i}\}$,
we have
\begin{align}
\sum_{i=1}^{n} u^{(new)}_i & \leq\sum_{i:\bv W_{ii}\neq1} u_{i}+\sum_{i:\bv W_{ii}=1}\tau_{i}^{\bv{S'A}}(\bv A)\nonumber \\
& \leq\frac{d}{\alpha}+\sum_{i:\bv W_{ii}=1}\tau_{i}^{\bv{S'A}}(\bv A)\nonumber \\
& =\frac{d}{\alpha}+\sum_{i:\bv W_{ii}=1}\tau_{i}^{\bv{S'A}}(\bv{WA}).\label{eq:sum_u_d_alpha}
\end{align}
Using $\bv W\preceq\bv I$, we have
\begin{equation}
\tau_{i}^{\bv{S'A}}(\bv{WA})\leq\tau_{i}^{\bv{S'WA}}(\bv{WA}).\label{eq:tau_sa_tau_swa}
\end{equation}
Now, note that $\bv S'= \sqrt{\alpha}\sqrt{\frac{3}{4}}\cdot \mathtt{Sample}\left(\bv u,9\alpha\right)= \sqrt{\alpha}\sqrt{\frac{3}{4}} \cdot \mathtt{Sample}\left(\alpha\bv u,9\right)$.
Since $\alpha\bv u$ is an overestimate of leverage scores for $\bv{WA}$,
Lemma \ref{leverage_score_sampling} (again with $\epsilon = 1/3$) shows that $\alpha \cdot \frac{1}{2}\bv A^{\top}\bv W^2\bv A \preceq \bv A^{\top}\bv{W}\bv S'^{2} \bv W \bv{A}$.
Hence \eqref{eq:tau_sa_tau_swa} along with Lemma \ref{leverage_score_approx} shows that
\[
\tau_{i}^{\bv{S'A}}(\bv{WA})\leq\frac{2}{\alpha}\tau_{i}(\bv{WA}).
\]
Combining with \eqref{eq:sum_u_d_alpha}, we have
\begin{align*}
\sum_{i=1}^{n} u^{(new)}_i & \leq\frac{d}{\alpha}+\frac{2}{\alpha}\sum_{i:\bv W_{ii}=1}\tau_{i}(\bv{WA})\\
& \leq\frac{d}{\alpha}+\frac{2d}{\alpha}\leq\frac{3d}{\alpha}.
\end{align*}
\end{proof} 

Choosing an undersampling rate $\alpha$ is equivalent to choosing a desired sampling rate and setting $\alpha$ accordingly. From this perspective, it is clear that the above theorem gives an extremely simple way to iteratively improve leverage scores. Start with $\bv{u}^{(1)}$ with $\norm{\bv{u}^{(1)}}_1 = s_1$. Undersample at rate $\frac{6d}{s_1}$ to obtain a sample of size $O(d \log d)$, which gives new leverage score estimates $\bv{u}^{(2)}$ with $\norm{\bv{u}^{(2)}}_1 = \frac{3d}{6d/s_1} = \frac{s_1}{2}$. Repeat this process, cutting the sum of leverage score estimates in half with each iteration. Recall that we restrict $\alpha < 1$, so once the sum of leverage score estimates converges on $O(d)$, this halving process halts -- as expected, we can not keep cutting the sum further.

\subsection{Improved Bound for Uniform Sampling}
The algorithm just described corresponds to Algorithm \ref{alg:algo3} in Section \ref{sec:algo} and differs somewhat from approaches discussed earlier (e.g. our proposed iterative algorithm from Section \ref{sec:intro}). It always maintains a sample of just $O(d\log d)$ rows that is improved iteratively.

Consider instead sampling few rows from $\bv{A}$ with the goal of estimating leverage scores well enough to obtain a spectral approximation with $n/2$ rows. In the uniform sampling case, when $\bv{u} = \bv{1}$, if we set $\alpha = \frac{d\log d}{6n}$ for example, then sampling $O(\alpha \norm{\bv{u}}_1 \log d) = O(d \log^2 d)$ rows uniformly will give us leverage score estimates summing to $\frac{n}{2\log d}$. This is good enough to cut our original matrix in half. However, we see that we have lost a $\log d$ factor to Theorem \ref{main_lemma}, which let us cut down to \emph{expected} size $\frac{n}{2}$ by sampling just $O(d \log d)$ rows uniformly.

At least when $\bv{u} = 1$, this $\log d$ factor can be eliminated. In Theorem \ref{leverage_score_undersampling}, we set $\bv S'=\sqrt{\alpha}\sqrt{\frac{3}{4}} \cdot \mathtt{Sample}\left(\bv u,9\alpha\right)$, meaning that rows selected for $\bv S'$ are included with weight $\sqrt{\alpha}\sqrt{\frac{3}{4}} \cdot \frac{1}{\sqrt{p_i}} = \sqrt{\frac{3\alpha}{4 \cdot \min \{1, 9 \alpha c \log d\}}}$. Instead of reweighting rows, consider simply setting all non-zero values in $\bv{S'}$ to be $1$. We know that our leverage score estimates will still be overestimates as we still have $\bv{S'} \preceq \bv{I}$ and so $\bv{A^\top} \bv{S'}^2 \bv{A} \preceq \bv{A^\top A}$. Further

Formally, consider two cases:
\begin{enumerate}
\item $(1 \le 9 \alpha c \log d)$. In this case, $\bv{S' A}$ is simply $\bv{A}$ itself, so we know our leverage score estimates are exact and thus their sum is $\leq d$. We can use them to obtain a spectral approximation with $O(d\log d)$ rows.

\item $(1 > 9 \alpha c \log d)$. In this case, we reweight rows by $\sqrt{\alpha}\sqrt{\frac{3}{4}} \cdot \frac{1}{\sqrt{p_i}} = \sqrt{\frac{3\alpha}{4 \cdot 9 \alpha c \log d}} = \sqrt{\frac{3}{4 \cdot 9 c\log d}}$. Thus, increasing weights in $\bv{S'}$ to $1$ will reduce leverage score estimates by a factor of $\frac{3}{4 \cdot 9 c\log d}$. So overall we have: 
\begin{align*}
\sum_{i=1}^{n} u^{(new)}_i &\leq\sum_{i:\bv W_{ii}\neq1} u_{i}+\sum_{i:\bv W_{ii}=1}\tau_{i}^{\bv{S'A}}(\bv A)\\
&\le \left|\left\{i:\bv W_{ii}\neq 1\right\} \right| + \frac{3}{4 \cdot 9 c\log d} \cdot \frac{2d}{\alpha}.
\end{align*}

Recall from Lemma \ref{leverage_score_sampling} that sampling by $\bv{u}^{(new)}$ actually gives a matrix with $\sum_i \min \{1, u_i^{(new)} \cdot \epsilon^{-2} c\log d \}$ rows. Thus, we obtain a $\frac{1+\epsilon}{1-\epsilon}$-spectral approximation to $\bv{A}$ with the following number of rows:
\begin{align*}
\left|\left\{i:\bv W_{ii}\neq 1\right\} \right| + \epsilon^{-2} c \log d \cdot \frac{3}{4 \cdot 9 c\log d} \cdot \frac{2d}{\alpha} \le \frac{d}{\alpha} + \frac{\epsilon^{-2}d}{6\alpha}.
\end{align*}
\end{enumerate} 

Setting $\alpha = \frac{m}{n \cdot 9c \log d}$ for some $m \le n$ so that $\mathtt{Sample}(\bv{1},9\alpha)$ samples rows at rate $m/n$ yields the following theorem:


\begin{theorem}\label{uniform_stronger}
Given $\bv{A} \in \mathbb{R}^{n \times d}$, suppose we sample rows uniformly and independently at rate $\frac{m}{n}$, without reweighting, to obtain $\bv{SA}$. Computing $\tilde \tau_i = \min \{1, \tau_i^{\bv{SA}}(\bv{A}) \}$ for each row and resampling from $\bv{A}$ by these estimates will, with high probability, return a $\frac{1+\epsilon}{1-\epsilon}$-spectral approximation to $\bv{A}$ with at most $O(\frac{nd \log d \epsilon^{-2}}{m})$ rows.
\end{theorem}

Choosing $m = O(d\log d)$ allows us to find a spectral approximation of size $\frac{n}{2}$, as long as $O(d\log d) < n$. This matches the bound of Theorem \ref{main_lemma}, but holds with high probability.

\section{Applications to Row Sampling Algorithms}
\label{sec:algo}

As discussed in the introduction, Theorems \ref{main_lemma}, \ref{leverage_score_undersampling}, and \ref{uniform_stronger} immediately yield new,
extremely simple algorithms for spectral matrix approximation. For clarity, we initially present versions running in \emph{nearly} input-sparsity time. However, we later explain how our first algorithm can be modified with standard techniques to remove log factors, achieving input-sparsity time and thus matching state-of-the-art results \cite{clarkson2013low, meng2013, osnap}. Our algorithms rely solely on row sampling, which preserves matrix sparsity and structure, possibly improving space usage and runtime for intermediate system solves.

\subsection{Algorithm Descriptions}
The first algorithm we present, \algoname{Repeated Halving}, is a simple recursive procedure. We uniformly sample $\frac{n}{2}$ rows from $\bv{A}$ to obtain $\bv{A}'$. By Theorems \ref{main_lemma} and \ref{uniform_stronger}, estimating leverage scores of $\bv{A}$ with respect to this sample allows us to immediately find a spectral approximation to $\bv{A}$ with $O(d \log d)$ rows. Of course, $\bv{A}'$  is still large, so computing these estimates would be slow. Thus, we \emph{recursively} find a spectral approximation of $\bv{A}'$ and use this to compute the estimated leverage scores.


\begin{algorithm}[H]
\caption{\algoname{Repeated Halving}}

{\bf input}: $n \times d$ matrix $\bv{A}$\\
{\bf output}: spectral approximation $\bv{\tilde{A}}$ consisting of
$O(d\log{d})$ rescaled rows of $\bv{A}$

\begin{algorithmic}[1]
\label{alg:algo1}
\STATE{Uniformly sample $\frac{n}{2}$ rows of $\bv{A}$ to form $\bv{A}'$}
\STATE{If $\bv{A}'$ has $> O(d\log d)$ rows, \textbf{recursively} compute a spectral approximation $\tilde{\bv{A}}'$ of $\bv{A}'$}
\STATE{Compute approximate generalized leverage scores of $\bv{A}$
w.r.t. ${\bv{\tilde A}}'$
}
\STATE{Use these estimates to sample rows of $\bv{A}$ to form $\tilde{\bv{A}}$}
\RETURN{$\tilde{\bv{A}}$}
\end{algorithmic}
\end{algorithm}


The second algorithm, \algoname{Refinement Sampling}, makes critical use of Theorem \ref{leverage_score_undersampling}, which shows that, given a set of leverage score upper bounds, we can undersample by these estimates and still significantly improve their quality with each iteration. We start with all of our leverage score upper bounds set to $1$ so we have $\norm{\bs{\tilde \tau}}_1 = n$. In each iteration, we sample $O(d \log d)$ rows according to our  upper bounds, meaning that we undersample at rate $\alpha = O\left(\frac{d}{\norm{\bs{\tilde \tau}}_1}\right)$. By Theorem \ref{leverage_score_undersampling}, we cut $\norm{\tilde{ \bs{\tau}}}_1$ by a constant fraction in each iteration. Thus, within $\log(n)$ rounds, $\norm{\tilde{\bs{\tau}}}_1$ will be $O(d)$ and we can simply use our estimates to directly obtain a spectral approximation to $\bv{A}$ with $O(d\log d)$ rows. 

\begin{algorithm}[H]
\caption{\algoname{Refinement Sampling}}

{\bf input}: $n \times d$ matrix $\bv{A}$\\
{\bf output}: spectral approximation $\bv{\tilde{A}}$ consisting of
$O(d\log{d})$ rescaled rows of $\bv{A}$

\begin{algorithmic}[1]
\label{alg:algo3}
\STATE{Initialize a vector of leverage score upper bounds, $\tilde{\bs{\tau}}$, to $\bv{1}$}
\WHILE{$\norm{\tilde{\bs{\tau}}}_1 > O(d)$}
\STATE{Undersample $O(d \log d)$ rows of $\bv{A}$ with
probabilities proportional to $\bs{\tilde \tau}$ to form $\tilde{\bv{A}}$}
\STATE{Compute a vector $\bv{u}$ of approximate generalized leverage scores of $\bv{A}$
w.r.t. $\tilde{\bv{A}}$
}
\STATE{Set $\tilde{\tau}_i = \min(\tilde{\tau}_i, u_i)$}
\ENDWHILE
\STATE{Use the final $\bs{\tilde{\tau}}$ to sample $O(d \log d)$ rows from $\bv{A}$ to form $\tilde{\bv{A}}$}
\RETURN{$\tilde{\bv{A}}$}
\end{algorithmic}
\end{algorithm}
%




\subsection{Runtime Analysis}

In analyzing the runtimes of these algorithms, we assume $n = O(\poly(d))$, which is a reasonable assumption for any practical regression problem.\footnote{A simple method for handling even larger values of $n$ is outlined in \cite{peng_iterative}.}
Furthermore, we use the fact that a $d \times d$ system can be solved in time $d^\omega$, where $\omega$ is the matrix multiplication exponent. However, we emphasize that, depending on the structure and sparsity of $\bv{A}$, alternative system solving methods may yield faster results or runtimes with different trade offs. For example, if the rows of $\bv{A}$ are sparse, solving a system in $\bv{\tilde A}$, where $\bv{\tilde A}$ consists of $O(d\log d)$ rescaled rows from $\bv{A}$ may be accelerated by using iterative conjugate gradient, or other Krylov subspace methods (which can also avoid explicitly computing $\bv{\tilde A}^\top \bv{\tilde A}$) . It is best to think of $d^\omega$ as the runtime of the fastest available system solver in your domain, and the quoted runtimes as general guidelines that will change somewhat depending on exactly how the above algorithms are implemented.


First, we give an important primitive showing that
estimates of generalized leverage scores can be computed efficiently. 
Computing exact generalized leverage scores is slow and we only need constant factor approximations, which will only increase our sampling rates and hence number of rows sampled by a constant factor.

\begin{lemma}
\label{lem:leverageEstimate}
Given $\bv{B}$ containing $O(d\log d)$ rescaled rows of $\bv{A}$, for any $\theta > 0$,
it is possible to compute an estimate of $\bs{\tau}^{\bv{B}}(\bv{A})$, $\bs{\tilde \tau}$,
in $O(d^{\omega}\log d + \nnz(\bv{A}) \theta^{-1} )$ time such that, w.h.p. in $d$, for all $i$, ${\tilde{\tau}}_i \geq {\tau}_{i}^{\bv{B}}(\bv{A})$ and ${\tilde{\tau}}_i \leq d^{\theta} {\tau}_{i}^{\bv{B}}(\bv{A})$.
\end{lemma}

By setting $\theta = O(\frac{1}{\log d})$, we can obtain constant factor approximations to generalized leverage scores in time $O(d^\omega \log d + \nnz(\bv{A}) \log d)$.

\begin{proof}[Proof Sketch]
Lemma \ref{lem:leverageEstimate} follows from a standard technique that uses Johnson-Lindenstrauss
projections ~\cite{AvronT11,peng_iterative, Spielman:2008:GSE:1374376.1374456}. Presuming $\bv{a}_i \perp \ker(\bv{B})$, the general idea is to write $\bs{\tau}_{i}^{\bv{B}}(\bv{A}) = \bv{a}_i^\top (\bv{B^\top B})^+ \bv{a}_i = \norm{\bv{B}(\bv{B^\top B})^+ \bv{a}_i}_2^2$. If we instead compute $\norm{\bv{G}\bv{B}(\bv{B^\top B})^+ \bv{a}_i}_2^2$, where $\bv{G}$ is a random Gaussian matrix with $O(\theta^{-1})$ rows, then by the Johnson-Lindenstrauss Lemma, with high probability, the approximation will be within a $d^\theta$ factor of $\norm{\bv{B}(\bv{B^\top B})^+ \bv{a}_i}_2^2$ for all $i$ (See Lemma 4.5 of \cite{peng_iterative}). 

The naive approach requires multiplying every row by $(\bv{B^\top B})^+$, which has height $d$ and would thus incur cost $\nnz(\bv{A})d$. 
Computing $\bv{G B}(\bv{B^\top B})^+$ takes at most $O(\nnz(\bv{A}) \theta^{-1})$ time to compute $\bv{GB}$, and $O(d^\omega \log d)$ time to compute $(\bv{B^\top B})$ and invert it. It then takes less than time $O(d^\omega)$ to multiply these two matrices. With $\bv{G B}(\bv{B^\top B})^+$ in hand, we just need to multiply by each row in $\bv{A}$ to obtain our generalized leverage scores, which takes time $O(\nnz(\bv{A}) \theta^{-1})$. 
If we use an alternative system solver instead of explicitly computing $(\bv{B^\top B})^+$, the JL reduction means we only need to solve $O(\theta^{-1})$ systems in $\bv{B}$ to compute $\bv{G B}(\bv{B^\top B})^+$ (one for each row of $\bv{G}$).

A slight modification is needed to handle the case when $\bv{a}_i \not \perp \ker(\bv{B})$ -- we need to check whether the vector
has a component in the null-space of $\bv{B}$. There are a variety of ways to handle this detection. For example, we can choose a random gaussian vector $\bv{g}$ and compute $\bv{g} - (\bv{B^\top B})^+ \bv{B} \bv{g}$.
This gives a random vector in the null space of $\bv{B}$,
so computing its dot product with any row $\bv{a}_i$ will tell us (with probability 1)
whether $\bv{a}_i$ is orthogonal to the null space or not.
\end{proof}


With this primitive in place, we can analyze the runtime of our two algorithms. For simplicity, we just give  runtimes for computing a constant factor spectral approximation to $\bv{A}$. Such an approximation is sufficient for use as a preconditioner in iterative regression algorithms \cite{avron2010blendenpik, clarkson2013low, rokhlinTygert}. Furthermore, it allows us to compute leverage scores of $\bv{A}$ up to a constant factor, from which we can sample rows to directly obtain a $(1+\epsilon)$ approximation with $O(d \log d \epsilon^{-2})$ rows. By Lemma \ref{lem:leverageEstimate} the runtime of this final refinement is just $O(\nnz(\bv{A}) \log d + d^\omega \log d)$. 

\begin{lemma}
\label{lem:algo1}
\algoname{Repeated Halving} (Algorithm~\ref{alg:algo1}) runs
in $O( \nnz(\bv{A}) \log d + d^{\omega} \log(n/d)\log d)$ time, outputting
a matrix with $\tilde{\bv{A}}$ with $O(d \log{d})$ rows such that
$\bv{\tilde A}^\top \bv{\tilde A} \approx_{2} \bv{A}^\top \bv{A}$.
\end{lemma}

\begin{proof}
The proof is by induction -- it suffices to show that the work
done at the top level is $O(\nnz(\bv{A}) \log d + d^\omega \log d)$. At each of the $O(\log (n/d))$ levels of recursion, we cut our matrix in half uniformly so $\nnz(\bv{A})$ will also be cut approximately in half with high probability.

By Theorem \ref{uniform_stronger}, sampling by $\tau_i^{\bv{A}'}(\bv{A})$ allows us to obtain $\bv{\tilde A}$ with $O(d\log d)$ rows. If we instead use $\tilde{\bv{A}}'$, our estimated leverage scores increase by at most a constant factor (since $\bv{\tilde A}'$ is a constant factor spectral approximation to $\bv{A}'$). Furthermore, using Lemma \ref{lem:leverageEstimate} to approximate generalized leverage scores increases our estimates by another constant factor at most. Overall, $\bv{\tilde A}$ will have $O(d\log d)$ rows as desired and our runtime at the top level is just the runtime of estimating leverage scores from Lemma \ref{lem:leverageEstimate} -- $O(\nnz(\bv{A}) \log d + d^\omega \log d)$.

\end{proof}

\begin{lemma}
\label{lem:algo3}
\algoname{Refinement Sampling} (Algorithm~\ref{alg:algo3}) runs
in $O( \nnz(\bv{A}) \log(\frac{n}{d}) \log d + d^{\omega} \log(\frac{n}{d})\log d)$ time, outputting
a matrix with $\tilde{\bv{A}}$ with $O(d \log{d})$ rows such that
$\bv{\tilde A}^\top \bv{\tilde A} \approx_{2} \bv{A}^\top \bv{A}$.
\end{lemma}

\begin{proof}
The row count of $\tilde{\bv{A}}$ and the fact that it spectrally approximates $\bv{A}$ follows from the termination condition and Lemma \ref{leverage_score_sampling}.
By Lemma \ref{lem:leverageEstimate}, each iteration runs in $O(\nnz(\bv{A}) \log d + d^\omega \log d)$ time. Thus it suffices to show that the algorithm terminates after $O(\log (n/d))$ iterations. At each iteration, we undersample by a factor $\alpha = \frac{c \cdot d}{\norm{\bs{\tilde \tau}}}_1$ for some constant $c$. So by Theorem \ref{leverage_score_undersampling}, $\norm{\bs{\tilde \tau}}_1$ decreases to $\frac{3d}{\alpha} = \frac{3\norm{\bs{\tilde \tau}}_1}{c}$. Setting $c = 6$, we cut $\norm{\bs{\tilde \tau}}_1$ in half each time. Since we start with $\norm{\bs{\tilde \tau}}_1 = n$ and stop when $\norm{\bs{\tilde \tau}}_1 = O(d)$, we terminate in $O(\log(n/d))$ iterations.

\end{proof}


\subsection{Achieving Input Sparsity Time}

We briefly note that, using techniques from \cite{peng_iterative}, it is possible to remove the $\log d$ factor on the $\nnz(\bv{A})$ term to achieve true input-sparsity time with \algoname{Repeated Halving}. The general idea is that, instead of using Lemma \ref{lem:leverageEstimate} to estimate generalized leverage scores from up to a constant factor using $\bv{A'}$, we only estimate them up to a $d^{\theta}$ factor for some constant $0 < \theta < 1$. Using these rough estimates, we obtain $\bv{\tilde A}$ with $O(d^{1+\theta} \log d)$ rows. Then, for the rows in $\bv{\tilde A}$, we can again compute generalized leverage scores with respect to $\bv{A'}$, now up to constant factors, and reduce down to just $O(d \log d)$ rows. In total, each iteration will take time $O(\theta^{-1} \nnz(\bv{A}) + d^\omega \log d + d^{2+\theta} \log^2 d)$, so obtaining a constant factor approximation to $\bv{A}$ takes time $O(\theta^{-1} \nnz(\bv{A}) + d^\omega \log^2 d + d^{2+\theta} \log^3 d)$. Recall that we assume $n = \poly(d)$, so we have $\log(n/d) = O(\log d)$ iterations.

In order to obtain a $(1+\epsilon)$-spectral approximation with only $O(d\log d \epsilon^{-2})$ rows, we first obtain a constant factor approximation, $\bv{\tilde A}$, with $O(d \log d)$ rows. We then use leverage scores estimated with $\bv{\tilde A}$ to compute a $(1+\epsilon/2)$ approximation to $\bv{A}$ with $O(d^{1+\theta} \log d \epsilon^{-2})$ rows. Finally, we again use leverage scores estimated with $\bv{\tilde A}$ and Lemma \ref{lem:leverageEstimate} with $\theta = O(1/\log d)$ to a compute a $(1+\epsilon/2)$ approximation to this smaller matrix with only $O(d \log d \epsilon^{-2})$ rows. This takes total time $O(\theta^{-1} \nnz(\bv{A}) + d^\omega \log d + d^{2+\theta} \log^2 d \epsilon^{-2})$. The  $d^{2+\theta} \log^2 d \epsilon^{-2}$ comes from applying Lemma \ref{lem:leverageEstimate} to refine our second approximation, which has $O(d^{1+\theta} \log d \epsilon^{-2})$ rows and thus at most $O(d^{2+\theta} \log d \epsilon^{-2})$ nonzero entries. Overall, the technique yields:

\begin{lemma}
\label{lem:rowSampleGeneric}
Given any constant $0 < \theta \leq 1$, 
and any error $0 \le \epsilon < 1$\footnote{If $\epsilon < 1/\poly(d)$, then $O(d \log{d} \epsilon^{-2}) > n$, so we can trivially return $\bv{\tilde A} = \bv{A}$ in $O(\nnz(\bv{A}))$ time.}, w.h.p. in $d$ we can compute a matrix 
$\tilde{\bv{A}}$ with $O(d \log{d} \epsilon^{-2})$ rows such that
$\bv{\tilde A}^\top \bv{\tilde A} \approx_{1+\epsilon} {\bv{A}}^\top {\bv{A}}$ in
$O(\nnz(\bv{A}) + d^{\omega}\log^2 d + d^{2+\theta}\epsilon^{-2})$
time. 
\end{lemma}

As is standard, $\log d$ factors on the $d^{2+\theta}$ term are `hidden' as we can just slightly increase the value of $\theta$ to subsume them.
The full tradeoff parameterized by $\theta$
is:
\begin{align*}
O(\theta^{-1} \nnz(\bv{A}) + d^{\omega}\log^2 d + d^{2+\theta}(\log^3 d + \log^2d \epsilon^{-2})).
\end{align*}

\subsection{General Sampling Framework}
It is worth mentioning that Algorithms \ref{alg:algo1} and \ref{alg:algo3}
are two extremes on a
spectrum of algorithms between halving and refinement sampling.
Generically, the full space of algorithms can be summarized using the pseudocode in
Algorithm~\ref{alg:generic}. For notation, note that $\bv{A}$ always refers to our \emph{original} data matrix. $\bv{\hat{A}}$ is the data matrix currently being processed in the recursive call to Algorithm~\ref{alg:generic}.

\begin{algorithm}[H]
\caption{Generic Row Sampling Scheme}

{\bf input}: original $n \times d$ matrix $\bv{A}$, current
$\bar{n} \times d$ matrix $\hat{\bv{A}}$.\\
{\bf output}: approximation $\bv{\tilde{A}}$ consisting of $O(d\log{d})$ rescaled rows of $\bv{A}$

\begin{algorithmic}[1]
\STATE{Uniform sample $n_1$ rows of $\bv{\hat{A}}$, $\bv{A}_1$ \label{ln:firstSample}}
\STATE{Approximate $\bv{A}_1$ with a row sample, $\bv{A}_2$, \textbf{recursively} \label{ln:secondSample}}
\STATE{Estimate generalized leverage scores using $\bv{A}_2$ and use them to sample $n_3$ rows of either $\bv{\hat{A}}$ or $\bv{A}$ itself to obtain $\bv{A}_3$\label{ln:thirdSample}}
\STATE{Approximate $\bv{A}_3$ with a row sample, $\bv{A}_4$, \textbf{recursively} \label{ln:fourthSample}}
\RETURN{$\bv{A}_4$}
\end{algorithmic}
\label{alg:generic}
\end{algorithm}
Different choices for 
$n_1$ and $n_3$ lead to different algorithms.
Note that the last recursion to approximate $\bv{A}_4$ has error build up
incurred from sampling to create $\bv{A}_3$.
As a result, this generic scheme has error buildup, but it can be removed
by sampling w.r.t. $\bv{A}$ instead of $\bv{\hat{A}}$.

Note that if we choose $n_1 = O(d \log{d})$, we can simply set $\bv{A}_2 \leftarrow \bv{A}_1$,
and the first recursive call in Line~\ref{ln:secondSample} is not necessary.
Also, Theorem~\ref{leverage_score_undersampling} gives that, if we pick $n_1$
sufficiently large, $n_3$ can be bounded by $O(d\log{d})$.
This would then remove the last recursive call to compute $\bv{A}_4$.
Such modifications lead to head and tail recursive algorithms, as well
as a variety of intermediate forms:

\begin{enumerate}
\item Head recursive algorithm, $n_1 = n / 2$, giving Algorithm~\ref{alg:algo1} (\algoname{Repeated Halving}).
\item Tail recursive algorithm, $n_1 = d \log{d}$, $n_3 = \frac{n}{2}$, sampled w.r.t. $\bv{\hat A}$. At each step error compounds so setting error to $\frac{1}{\log{n}}$ per step gives a constant factor approximation.
\item $n_1 = d \log{d}$, $n_3 = d\log{d}$, sampled w.r.t. $\bv{A}$, giving Algorithm~\ref{alg:algo3} (\algoname{Refinement Sampling}).
\item For situations where iterations are expensive, e.g. MapReduce,
a useful choice of parameters is likely $n_1 = n_3 = O(\sqrt{nd\log{d}})$,
This allows one to compute $\bv{A}_2$ and $\bv{A}_4$ without recursion,
while still giving speedups.
\end{enumerate}

\section{Acknowledgements}
We would like to thank Jonathan Kelner and Michael Kapralov for helpful discussions.
This work was partially supported by NSF awards 0843915, 1111109, and 0835652, CCF-AF-0937274,  CCF-0939370,  and CCF-1217506, NSF Graduate Research Fellowship grant 1122374, Hong Kong RGC grant 2150701, AFOSR grant FA9550-13-1-0042, and the Defense Advanced Research Projects Agency (DARPA).

\bibliography{itcs}
\bibliographystyle{alpha}

\appendix
\section{Properties of Leverage Scores}
\label{sec:app:lever_scores}

\subsection{Spectral Approximation via Leverage Score Sampling}
\label{subsec:app:spectral_approx}
Here we prove Lemma \ref{leverage_score_sampling}, which states that it is possible to obtain a spectral approximation by sampling rows from $\bv{A}$ independently with probabilities proportional to leverage score overestimates.

\begin{replemma}{leverage_score_sampling}[Spectral Approximation via Leverage Score Sampling] 
Given an error parameter $0 < \epsilon < 1$, let $\bv u$ be a vector of leverage score overestimates, i.e., $\tau_{i}(\bv A)\le u_{i}$ for all $i$. Let $\alpha$ be a sampling rate parameter and let $c$ be a fixed positive constant. For each row, we define a sampling probability $p_i = \min \{1,\alpha\cdot u_i  c \log d\}$. Furthermore, we define a function $\mathtt{Sample}(\bv u,\alpha)$, which returns a random diagonal matrix $\bv{S}$ with independently chosen entries. $\bv{S}_{ii} = \frac{1}{\sqrt{p_i}}$ with probability $p_i$ and $0$ otherwise.

Setting $\alpha = \epsilon^{-2}$, $\bv S$ has at most $\sum_i \min \{1,\alpha\cdot u_i  c \log d \} \le \alpha c\log d\norm{\bv{u}}_1$
non-zero entries and $\frac{1}{\sqrt{1+\epsilon}} \bv{SA}$ is a ${\frac{1+\epsilon}{1-\epsilon}}$-spectral approximation for $\bv A$ with probability at least $1-d^{-c/3}$. 
\end{replemma}

We rely on the following matrix concentration result, which is a variant of Corollary 5.2 from \cite{tropp2012user}, given by Harvey in \cite{harveyTalk}:
\begin{lemma}\label{tropp_bound}
Let $\bv{Y}_1...\bv{Y}_k$ be independent random positive semidefinite matrices of size $d \times d$. Let $\bv{Y} = \sum_i^k \bv{Y}_i$ and let $\bv{Z} = \E[\bv{Y}]$. If $\bv{Y}_i \preceq R \cdot \bv{Z}$ then
\begin{align*}
\Pr \left[\sum_i^k \bv{Y}_i \preceq (1-\epsilon) \bv{Z}\right] \le d e^{\frac{-\epsilon^2}{2R}},
\end{align*}
and 
\begin{align*}
\Pr \left[\sum_i^k \bv{Y}_i \succeq (1+ \epsilon) \bv{Z}\right] \le d e^{\frac{-\epsilon^2}{3R}}.
\end{align*}
\end{lemma}

\begin{proof}[Proof of Lemma \ref{leverage_score_sampling}]
For each row $\bv{a}_i$ of $\bv{A}$ choose $\bv{Y}_i = \frac{\bv{a}_i\bv{a}_i^\top}{p_i}$ with probability $p_i$, and $0$ otherwise. So $(\bv{S} \bv{A})^\top (\bv{S} \bv{A}) = \sum_i \bv{Y}_i$. Note that $\bv{Z} = \E[\sum_i \bv{Y}_i] = \sum_i \E[\bv{Y}_i] = \sum_i \bv{a}_i\bv{a}_i^\top = \bv{A}^\top\bv{A}$, as desired. To apply Lemma \ref{tropp_bound}, we will want to show that
\begin{align}
\label{need_to_show_r_bound}
\forall i, \bv{Y}_i \preceq \frac{1}{c \log d \epsilon^{-2}} \cdot \bv{A}^\top\bv{A}.
\end{align}
First, consider when $p_i < 1$. The $p_i = 1$ is slightly less direct, and we will deal with it shortly. $\alpha = \epsilon^{-2} \ge 1$, so $p_i < 1$ implies that $\tau_i(\bv{A}) \le u_i \le \frac{1}{c \log d}$. We have
\begin{align}\label{matrix_norm_bound}
\bv{Y}_i \preceq \frac{\bv{a}_i\bv{a}_i^\top}{u_i \cdot c \log d \epsilon^{-2}} \preceq \frac{\bv{a}_i\bv{a}_i^\top}{ \tau_i(\bv{A}) \cdot c \log d \epsilon^{-2}} \preceq \frac{1}{c \log d \epsilon^{-2}} \cdot \bv{A^\top A},
\end{align}
since 
\begin{align}\label{basic_leverage_bound}
\frac{\bv{a}_i\bv{a}_i^\top}{\tau_i(\bv{A})} \preceq \bv{A^\top A}.
\end{align}
Equation (\ref{basic_leverage_bound}) is proved by showing that, for all $\bv{x} \in \mathbb{R}^d$, $\bv{x}^\top \bv{a}_i\bv{a}_i^\top \bv{x} \le\tau_i(\bv{A}) \cdot \bv{x}^\top \bv{A^\top A} \bv{x}$. We can assume without loss of generality that $\bv{x}$ is in the column space of $\bv{A^\top A}$ since, letting $\bv{x'}$ be the component of $\bv{x}$ in the null space of $(\bv{A^\top A})$, $0 = \bv{x'}^\top (\bv{A^\top A}) \bv{x'} = \bv{x'}^\top \bv{a}_i\bv{a}_i^\top \bv{x'}$. Now, with $\bv{x}$ in the column space, for some $\bv{y}$ we can write $\bv{x} = \bv{(\bv{A^\top A})}^{+/2} \bv{y}$ where $(\bv{A^\top A})^{+/2} = \bv{V} \bv{\Sigma}^{-1} \bv{V}^\top$ (recall that $(\bv{A^\top A})^+ = \bv{V} \bv{\Sigma}^{-2} \bv{V}^\top$ if the SVD of $\bv{A}$ is given by $\bv{A} = \bv{U\Sigma V^\top}$). So now we consider
\begin{align*}
\bv{y}^\top (\bv{A^\top A})^{+/2} \bv{a}_i\bv{a}_i^\top (\bv{A^\top A})^{+/2} \bv{y}.
\end{align*}
$(\bv{A^\top A})^{+/2} \bv{a}_i\bv{a}_i^\top (\bv{A^\top A})^{+/2}$ is rank 1, so its maximum (only) eigenvalue is equal to its trace.
By the cyclic property, $\tr((\bv{A^\top A})^{+/2} \bv{a}_i\bv{a}_i^\top (\bv{A^\top A})^{+/2}) = \tr(\bv{a}_i^\top (\bv{A^\top A})^{+} \bv{a}_i) = \tau_i(\bv{A})$. Furthermore, the matrix is positive semidefinite, so 
\begin{align*}
\bv{y}^\top (\bv{A^\top A})^{+/2} \bv{a}_i\bv{a}_i^\top (\bv{A^\top A})^{+/2} \bv{y} \le \tau_i(\bv{A}) \norm{\bv{y}}_2^2 = \tau_i(\bv{A}) \cdot \bv{y}^\top \bv{(A^\top A)}^{+/2} \bv{(A^\top A)} \bv{(A^\top A)}^{+/2} \bv{y},
\end{align*}
which gives us \eqref{basic_leverage_bound} and thus (\ref{need_to_show_r_bound}) when $p_i < 1$.

Equation (\ref{need_to_show_r_bound}) does not hold directly when $p_i = 1$. In this case, $\bv{Y}_i = \bv{a}_i \bv{a}_i^\top$ with probability $1$. However, selecting $\bv{Y}_i$ is exactly the same as selecting and summing $c \log d \epsilon^{-2}$ random variables $\bv{Y}^{(1)}_i,...,\bv{Y}^{(c \log d \epsilon^{-2})}_i$, each equal to $\frac{1}{c \log d \epsilon^{-2}} \cdot \bv{a}_i \bv{a}_i^\top$ with probability $1$, and thus clearly satisfying 
\begin{align}\label{matrix_norm_bound2}
\bv{Y}^{(j)} \preceq \frac{1}{c \log d \epsilon^{-2}} \cdot \bv{A^\top}\bv{A}.
\end{align}
We can symbolically replace $\bv{Y}_i$ in our Lemma \ref{tropp_bound} sums with these smaller random variables, which does not change $\bv{Z} = \E[\bv{Y}]$, but proves concentration. We conclude that:
\begin{align*}
(1-\epsilon) \cdot \bv{A^\top A} \preceq \sum_i \bv{Y}_i \preceq (1+\epsilon) \cdot \bv{A^\top A}
\end{align*}
with probability at least
\begin{align*}
1 - de^{\frac{-c \log d \epsilon^{-2} \epsilon^2}{3}} \geq 1 - d^{1-c/3}.
\end{align*}

As noted, $\bv{\tilde A^\top \tilde A} = \bv{A}^\top \bv{S}^2 \bv{A} =  \sum_i \bv{Y}_i$, so this gives us that $\frac{1}{\sqrt{1+\epsilon}} \bv{SA}$ is a $\frac{1+\epsilon}{1-\epsilon}$-spectral approximation to $\bv{A}$ with high probability.
Furthermore, by a standard Chernoff bound, $\bv{S}$ has $\sum_i \min \{1, u_i \cdot \alpha c\log d \} \le \alpha c\log d\norm{\bv{u}}_1 $
nonzero entries with high probability.
\end{proof}

\subsection{Rank 1 Updates}

Here we prove Lemma~\ref{rank_one_updates}, making critical use of the Sherman-Morrison formula for the Moore-Penrose pseudoinverse \cite[Thm 3]{meyer1973generalized}.

\begin{replemma}{rank_one_updates}[Leverage Score Changes Under Rank 1 Updates] Given any
$\bv A \in \R^{n\times d}$, $\gamma\in(0,1)$, and $i\in[n]$, let $\bv W$ be a diagonal
matrix such that $\bv W_{ii}=\sqrt{1-\gamma}$ and $\bv W_{jj}=1$
for all $j\neq i$. Then, we have 
\begin{align*}
\tau_{i}(\bv{WA}) & =\frac{(1-\gamma)\tau_{i}(\bv A)}{1-\gamma\tau_{i}(\bv A)}\le\tau_{i}(\bv A),
\end{align*}
and for all $j\neq i$,
\begin{align*}
\tau_{j}(\bv{WA}) & =\tau_{j}(\bv A)+\frac{\gamma\tau_{ij}(\bv A)^{2}}{1-\gamma\tau_{i}(\bv A)}\ge\tau_{j}(\bv A).
\end{align*}
\end{replemma}

\begin{proof}
\begin{align*}
\tau_{i}(\bv{WA}) & =\bv{\mathbbm{1}}_{i}\bv{WA}\left(\bv{A}^{\top}\bv {W}^{2}\bv {A}\right)^{+}\bv{A^{\top}W^{\top}}\bv{\mathbbm{1}_{i}}^{\top}\tag{definition of leverage scores}\\
 & =(1-\gamma)\bv a_{i}^{\top}\left(\bv{A^{\top}A}-\gamma\bv a_{i}\bv a_{i}^{\top}\right)^{+}\bv a_{i}\tag{definition of \ensuremath{\bv W}}\\
 & =(1-\gamma)\bv a_{i}^{\top}\left((\bv{A^{\top}A})^{+}+\gamma\frac{(\bv{A^{\top}A})^{+}\bv a_{i}\bv a_{i}^{\top}(\bv{A^{\top}A})^{+}}{1-\gamma\bv a_{i}^{\top}(\bv{A^{\top}A})^{+}\bv a_{i}}\right)\bv a_{i}\tag{Sherman-Morrison formula}\\
 & =(1-\gamma)\left(\tau_{i}(\bv A)+\frac{\gamma\tau_{i}(\bv A)^{2}}{1-\gamma\tau_{i}(\bv A)}\right)\\
 & =\frac{(1-\gamma)\tau_{i}(\bv A)}{1-\gamma\tau_{i}(\bv A)}\\
 & \le\tau_{i}(\bv A).
\end{align*}
Similarly,
\begin{align*}
\tau_{j}(\bv{WA}) & =\bv a_{j}^{\top}\left(\bv{A^{\top}A}-\gamma\bv a_{i}\bv a_{i}^{\top}\right)^{+}\bv a_{j}\\
 & =\bv a_{j}^{\top}\left((\bv{A^{\top}A})^{+}+\gamma\frac{(\bv{A^{\top}A})^{+}\bv a_{i}\bv a_{i}^{\top}(\bv{A^{\top}A})^{+}}{1-\gamma\bv a_{i}^{\top}(\bv{A^{\top}A})^{+}\bv a_{i}}\right)\bv a_{j}\\
 & =\tau_{j}(\bv A)+\frac{\gamma\tau_{ij}(\bv A)^{2}}{1-\gamma\tau_{i}(\bv A)}\\
 & \ge\tau_{j}(\bv A).
\end{align*}
\end{proof}

\subsection{Lower Semi-continuity of Leverage Scores}

Here we prove Lemma~\ref{lem:lower_semi_tau} by providing a fairly general inequality, Lemma \ref{lem:lever_compare}, for relating leverage scores under one set of weights to leverage scores under another. 
\begin{replemma}{lem:lower_semi_tau}[Leverage Scores are Lower Semi-continuous]
$\bs{\tau}(\bv{W}\bv{A})$
is lower semi-continuous in the diagonal matrix $\bv W$, i.e. for any sequence $\bv{W}^{(k)} \rightarrow \overline{\bv W}$
with $\bv W_{ii}^{(k)} \geq 0$ for all $k$ and $i$, we have
\begin{equation}
\label{eq:lower_semi_tau}
\tau_{i}(\overline{\bv W}\bv A)\leq\liminf_{k\rightarrow\infty}\tau_{i}(\bv W^{(k)}\bv A).
\end{equation}
\end{replemma}

\begin{lemma}[Comparing Leverage Scores]
\label{lem:lever_compare}
Let $\bv{W}, \bv{\overline{W}} \in \R^{n \times n}$ be non-negative diagonal matrices and suppose that $\bv{W}_{ii} > 0$ and $\bv{\overline{W}}_{ii} > 0$ for some $i \in [n]$. Then
\begin{equation}
\label{eq:comparing_leverage_scores}
\tau_i(\overline{\bv W}\bv A)
\leq
\frac{\bv{\overline{W}}_{ii}^2}{\bv{W}_{ii}^2}
\left(
1 + \sqrt{\lambda_{\max}\left(\bv A\left(\bv A^{\top}\overline{\bv W}^{2}\bv A\right)^{+}\bv A^{\top}\right)}
\norm{\bv{W} - \overline{\bv{W}}}_\infty
\right)^{2}
\tau_{i}(\bv{WA}).
\end{equation} 
\end{lemma}

\begin{proof}
Scaling the variables in Lemma~\ref{leverage_score_opt} we have that there exists $\bv{x} \in \R^{d}$ such that 
\begin{equation}
\label{lem:lever_compare_1}
\bv{A}^{\top} \bv{W} \bv{x} = \bv{a}_{i}
\enspace\text{ and }\enspace
\norm{\bv{x}}_2^2 = \frac{\tau_{i}(\bv{W} \bv{A})}{\bv{W}_{ii}^{2}}.
\end{equation}
Note that $\bv A^\top\left(\bv{W} - \overline{\bv W}\right)\bv{x}$ is
in the image of $\bv A^\top\overline{\bv W}$ as $\bv A^\top\bv W\bv x=\bv a_{i}$
and $\bv{\overline{W}}_{ii} \neq 0$. Consequently,
\[
\bv A^\top\overline{\bv W}\bv y=\bv A^\top\left(\bv{W} - \overline{\bv W}\right)\bv{x}^{(k)}
\enspace
\text{ for }
\enspace
\bv{y}
\eqdef
\overline{\bv W} \bv A
\left(\bv A^\top\overline{\bv W}^{2}\bv A\right)^{+}\bv A^\top\left(\bv{W} -\overline{\bv W}\right)\bv{x}.
\]
Since $\bv A^\top\overline{\bv W}(\bv x+\bv y)=\bv a_{i}$, Lemma~\ref{leverage_score_opt} implies
\begin{equation}
\label{lem:lever_compare_2}
\tau_{i}(\overline{\bv W}\bv A) \leq \overline{\bv W}_{ii}^{2}\left\Vert \bv x+\bv y\right\Vert _{2}^{2}.
\end{equation}
We can bound the contribution of $\bv y$ by 
\begin{eqnarray}
\left\Vert \bv y\right\Vert _{2}^{2} & \leq & \left\Vert \overline{\bv W}\bv A\left(\bv A^\top\overline{\bv W}^{2}\bv A\right)^{+}\bv A^\top\left(\bv{W} - \overline{\bv W}\right)\bv{x} \right\Vert _{2}^{2}
\nonumber
\\
 & \leq & \lambda_{\max}\left(\bv A\left(\bv A^\top\overline{\bv W}^{2}\bv A\right)^{+}\bv A^\top\right)\left\Vert \bv{W} -\overline{\bv{W}}\right\Vert _{\infty}^{2} \norm{\bv{x}}_{2}^{2}.
 \label{lem:lever_compare_3}
\end{eqnarray}
Applying triangle inequality to \eqref{lem:lever_compare_1}, \eqref{lem:lever_compare_2}, and \eqref{lem:lever_compare_3} yields the result.
\end{proof}


\begin{proof}[Proof of Lemma~\ref{lem:lower_semi_tau}]
For any $i \in [n]$ such that $\overline{\bv W}_{ii} = 0$  \eqref{eq:lower_semi_tau} follows trivially from the fact that leverage scores are non-negative. For any $i \in [n]$ such that $\overline{\bv W}_{ii} > 0$, since $\bv{W}^{(k)} \rightarrow \bv{\overline{W}}$ we know that, for all sufficiently large $k \geq N$ for some fixed value $N$, it is the case that $\bv{W}^{(k)}_{ii} > 0$. Furthermore, this implies that as $k \rightarrow \infty$ we have $\bv{\overline{W}}_{ii}^{2}/(\bv{W}_{ii}^{(k)})^{2} \rightarrow 1$
and $\left\Vert \bv W^{(k)}-\overline{\bv W}\right\Vert _{\infty}\rightarrow 0$. Applying Lemma~\ref{lem:lever_compare} with $\bv{W} = \bv{W}^{(k)}$ and taking $\liminf_{k\rightarrow\infty}$ on both sides of \eqref{eq:comparing_leverage_scores} gives the result.
\end{proof}

\end{document}